\documentclass[final]{siamltex}
\usepackage{amsthm}
\usepackage{amsmath,graphicx,url,dbnsymb}
\usepackage{stmaryrd}  
\usepackage{pat}

\newcommand{\eps}{\varepsilon}
\newcommand{\upen}[1]{\llceil#1\rrceil}

\title{Robust Geometric Spanners\thanks{A preliminary version of this
   paper appears in the proceedings of the \emph{29th ACM Symposium on
   Computational Geometry (SoCG~2013)}.}}
\author{Prosenjit Bose\thanks{School of Computer Science, Carleton University, 1125 Colonel By Drive, Ottawa, CANADA, K1S~5B6 ({\tt \{jit,morin,michiel\}@scs.carleton.ca}).} \and
        Vida Dujmovi\'c\thanks{School of Mathematics and Statistics and Department of Systems and Computer Engineering, Carleton University, 1125 Colonel By Drive, Ottawa, CANADA, K1S~5B6 ({\tt \{jit,morin,michiel\}@scs.carleton.ca}).}\and
        Pat Morin\footnotemark[2] \and
        Michiel Smid\footnotemark[2]}

\begin{document}
\maketitle

\begin{abstract}
  Highly connected and yet sparse graphs (such as expanders or graphs
  of high treewidth) are fundamental, widely applicable and extensively
  studied combinatorial objects.  We initiate the study of such highly
  connected graphs that are, in addition, geometric spanners.  We define
  a property of spanners called robustness.  Informally, when one removes
  a few vertices from a robust spanner, this harms only a small number of
  other vertices.  We show that robust spanners must have a superlinear
  number of edges, even in one dimension.  On the positive side, we give
  constructions, for any dimension, of robust spanners with a near-linear
  number of edges.
\end{abstract}

\begin{keywords}
spanners, stretch-factor, spanning-ratio, tree-width, connectivity, expansion
\end{keywords}

\begin{AMS}
68M10, 05C10, 65D18
\end{AMS}

\pagestyle{myheadings}
\thispagestyle{plain}
\markboth{P. BOSE, V.~DUJMOVI\'C, P. MORIN, AND M. SMID}{ROBUST GEOMETRIC SPANNERS}

\section{Introduction} 
The cost of building a network, such as a computer network or a network
of roads, is closely related to the number of edges in the underlying
graph that models this network.  This gives rise to the requirement
that this graph be sparse.  However, sparseness typically has to be
counter-balanced with other desirable graph (that is, network design)
properties such as reliability and efficiency.

The classical notion of graph connectivity provides some guarantee of
reliability. In particular, an $r$-connected graph remains connected as
long as fewer than $r$ vertices are removed. However these graphs are
not sparse for large values of $r$;  an $r$-connected graph with $n$
vertices has at least $rn/2$ edges.

For many applications, disconnecting a small number of nodes from the
network is an inconvenience for the nodes that are disconnected, but has
little effect on the rest of the network.  In contrast, disconnecting
a large part (say, a constant fraction) of the network from the rest
is catastrophic.
For example, it may be tolerable that the failure of one network component
cuts off internet access for the residents of a small village. However,
the failure of a single component that eliminates all communications
between North America and Europe would be disastrous.

This global notion of connectivity is captured in graph theory by
expanders and graphs of high treewidth, each of which can have a linear
number of edges.  These two properties of graphs have an enormous number
of applications and have been the subject of intensive research for
decades.  See, for example, the book by Kloks \cite{k94} or the surveys
by Bodlaender \cite{b98,b07} on treewidth and the survey by Hoory, Linial,
and Wigderson \cite{hlw06} on expanders.

In this paper, we consider how to combine this global notion of
connectivity with another desirable property of geometric graphs:
low spanning ratio (a.k.a., low stretch factor or low dilation),  the
property of approximately preserving Euclidean distances between vertices.
In particular, given a set of $n$ points in $\R^d$, we study the problem
of constructing a graph on these points where the weights of the edges
are given by the Euclidean distance between their endpoints. We wish
to construct a graph such that
\begin{enumerate}
  \item The graph is sparse: the graph has $o(n^2)$ edges
  \item The graph is a spanner: (weighted) shortest paths in the graph
    do not exceed the Euclidean distance between their endpoints by more
    than a constant factor; and
  \item The graph has high global connectivity: removing a small number
    of vertices leaves a graph in which a set of vertices of size $n-o(n)$
    are all in the same component and all vertices in this set have
    spanning paths between them.
\end{enumerate}

This is the first paper to consider combining low spanning ratio
with high global connectivity.  This is somewhat surprising, since
many variations on sparse geometric spanners have been studied,
including spanners of low degree \cite{abcghsv08,cc10,s06}, spanners
of low weight \cite{bcfms10,dn97,gln02}, spanners of low diameter
\cite{ams94,ams99}, planar spanners \cite{accdsz96,c89,dj89,kg89},
spanners of low chromatic number \cite{bccmsz09}, fault-tolerant
spanners \cite{abfg09,cz04,lns02,l99}, low-power spanners
\cite{aack11,ss10,wl06}, kinetic spanners \cite{ab11,abg10},
angle-constrained spanners \cite{cs10}, and combinations of these
\cite{admss95,as97,bfrv12,bgs05,bsx09,cc10b}.  The closest related
work is that on fault-tolerant spanners \cite{abfg09,cz04,lns02,l99},
but $r$-fault-tolerance is analogous to the traditional definition
of $r$-connectivity in graph theory and suffers the same shortcoming:
every $r$-fault-tolerant spanner has $\Omega(rn)$ edges.

In the next few subsections, we formally define robust spanners
and discuss, at a more rigorous level, the relationship between
robust-spanners, fault-tolerant spanners, and expanders.  From this point
onwards, all graphs we discuss have vertices that are points in $\R^d$;
$n$ refers to the number points/vertices;  all distances between pairs of
points are Euclidean distances; and any shortest path in a graph refers
to the shortest (Euclidean) path that uses only edges of the graph.

\subsection{Robustness}

Let $V\subset \R^d$ be a set of $n$ points in $\R^d$.  An undirected
graph $G=(V,E)$ is a (geometric) \emph{$t$-spanner of $V^-\subseteq V$}
if, for every pair $x,y\in V^-$,
\[
  \frac{\|xy\|_G}{\|xy\|} \le t \enspace ,
\]
where $\|xy\|$ denotes the Euclidean distance between $x$ and $y$ and
$\|xy\|_G$ denotes the length of the Euclidean shortest path from $x$
to $y$ that uses only edges in $G$.  Here we use the convention that
$\|xy\|_G=\infty$ if there is no path, in $G$, from $x$ to $y$.  We say
simply that $G$ is a \emph{$t$-spanner} if it is a $t$-spanner of $V$
(i.e., $V^-=V$).  We point out that, although $t$ is always at least 1,
it need not be an integer.

Geometric $t$-spanners have been studied extensively and have applications
in robotics, graph theory, data structures, wireless networks, and network
design.  A book \cite{ns07} and handbook chapter \cite{e99} provide
extensive discussions of geometric $t$-spanners and their applications.

For a graph $G=(V,E)$ and a subset $S\subseteq V$ of $G$'s vertices, we
denote by $G\setminus S$ the subgraph of $G$ induced by $V\setminus S$.
A graph $G$ is an \emph{$f(k)$-robust $t$-spanner} of $V$ if, for every
subset $S\subseteq V$, there exists a superset $S^+\supseteq S$, $|S^+|\le
f(|S|)$, such that $G\setminus S$ is a $t$-spanner of $V\setminus S^+$.

An example is shown in \figref{grid} which suggests that the
$\sqrt{n}\times\sqrt{n}$ grid graph is an $O(k^2)$-robust 3-spanner.
The set $S^+$ in this example is obtained by choosing ``disjoint''
squares that cover the vertices of $S$ and adding to $S^+$ any vertices
contained in these squares.  A short path between any two vertices in
$V\setminus S^+$ is obtained by starting with some shortest path in $G$
between these two vertices and then routing around any of the square
holes encountered by this path. (A proof that the grid graph is indeed
an $O(k^2)$-robust 3-spanner is sketched in \secref{summary}.)

\begin{figure}
  \begin{center}
    \includegraphics{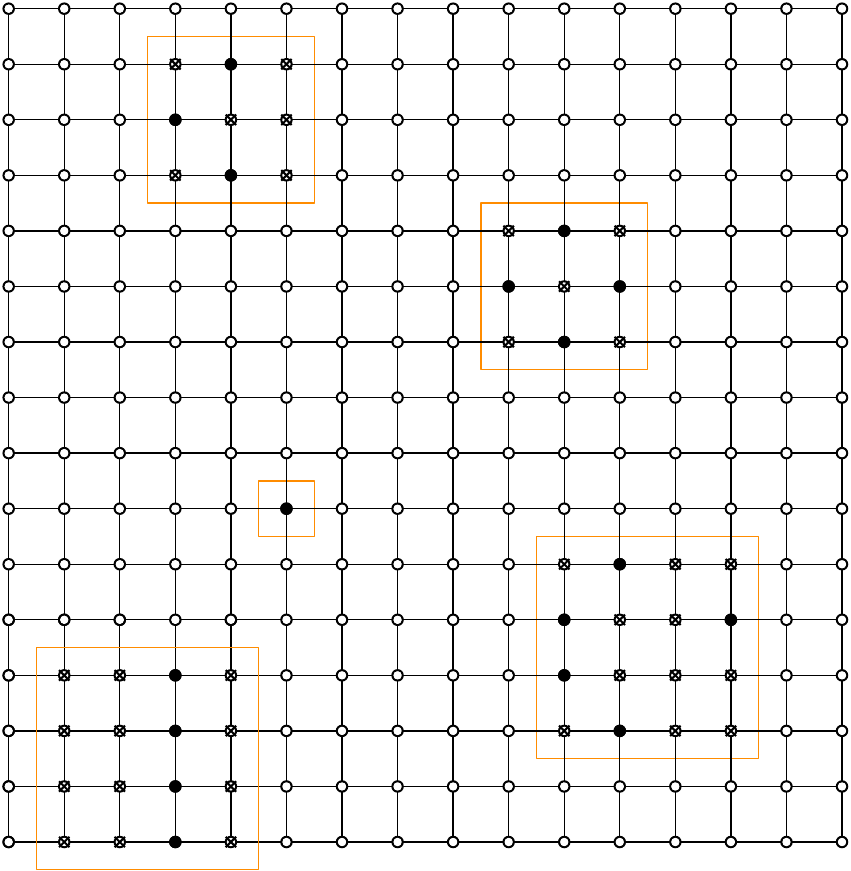}
  \end{center}
  \caption{From the set $S$ (whose elements are denoted by \textbullet)
  we find a superset $S^+$ (whose elements are denoted by $\times$ and
  \textbullet) so that $G\setminus S$ is a 3-spanner of $G\setminus S^+$
  (whose vertices are denoted by $\circ$).}
  \figlabel{grid}
\end{figure}

One can think of an $f(k)$-robust $t$-spanner in terms of network
reliability.  If a network is an $f(k)$-robust $t$-spanner, and $k$ nodes
of the network fail, then the network remains a $t$-spanner of $n-f(k)$
of its nodes.  Intuitively, most of the network survives the removal of
$k$ nodes, provided that $k$ is small enough that $f(k)\ll n$.

A slightly stronger version of robustness, which is achieved by some
of our constructions, requires that $G\setminus S^+$ induces a spanner.
Under this definition, the graph $G\setminus S^+$ must be a $t$-spanner of
$V\setminus S^+$. For example, the grid graph in \figref{grid} satisfies
this stronger definition since the vertices inside the squares are not
used in the short paths between vertices outside the squares.  In some
applications, this stronger definition may be preferable since the nodes
in $S^+$, which no longer gain the full benefits of the network $G$,
are not required to help with the routing of messages between nodes of
$V\setminus S^+$.  (An open problem related to this stronger definition
of robustness is discussed in \secref{summary}.)

\subsection{Robustness versus Fault-Tolerance}

Robustness is related to, but different from, $r$-fault tolerance.
An \emph{$r$-fault-tolerant} $t$-spanner, $G=(V,E)$, has the property
that $G\setminus S$ is a $t$-spanner of $V\setminus S$ for any subset
$S\subseteq V$ of size at most $r$.  In our terminology, an $r$-fault
tolerant spanner is $f(k)$-robust with
\[
    f(k) = \begin{cases}k & \text{for $k \le r$}  \\
                   \infty & \text{for $k > r$.}  \\
   \end{cases}
\]

At a minimum, an $r$-fault-tolerant spanner must remain connected
after the removal of any $r$ vertices.  This immediately implies
that any $r$-fault-tolerant spanner with $n>r$ vertices has at least
$(r+1)n/2$ edges, since every vertex must have degree at least $r+1$.
Several constructions of $r$-fault-tolerant spanners with $O(rn)$ edges
exist \cite{cz04,lns02,l99}.

In contrast, surprisingly sparse $f(k)$-robust $t$-spanners exist.
For example, we show that for one-dimensional point sets, there exists
$O(k\log k)$-robust 1-spanners with $O(n\log n)$ edges; the removal of
any set of $o(n/\log n)$ vertices leaves a subgraph of size $n-o(n)$ that
is a $1$-spanner.  An $r$-fault-tolerant spanner with $r=n/\log n$ also
has this property, but all such graphs have $\Omega(n^2/\log n)$ edges.

We suggest that in many applications where an $r$-fault-tolerant spanner
is used, an $f(k)$-robust spanner may be a better choice.  For example,
one might build an $r$-fault-tolerant spanner so that a network survives
up to $r$ faults, perhaps because more than $r$ faults is viewed as
unlikely.  Using an $f(k)$-robust spanner instead means that, if $r'\le
r$ faults do occur, then an additional $f(r')-r'$ nodes suffer, but the
remaining $n-f(r')$ nodes are unaffected.  In one case, the network
loses $r'$ nodes while in the other case $f(r')$ nodes are affected.
For slow-growing functions $f$ this may be perfectly acceptable.

The use of an $f(k)$-robust spanner in place of an $r$-fault-tolerant
spanner has the additional advantage that the maximum number of faults
need not be known in advance.  In the unlikely event that $r'> r$
faults occur, the network continues to remain usable.  In particular,
after $r'>r$ faults, the usable network has size at least $n-f(r')$. In
contrast, even with $r'=r+1$ faults, an $r$-fault-tolerant spanner may
have no component of size larger than $n/2$; see \figref{rft-problem}
for an example.

\begin{figure}
  \begin{center}
    \includegraphics{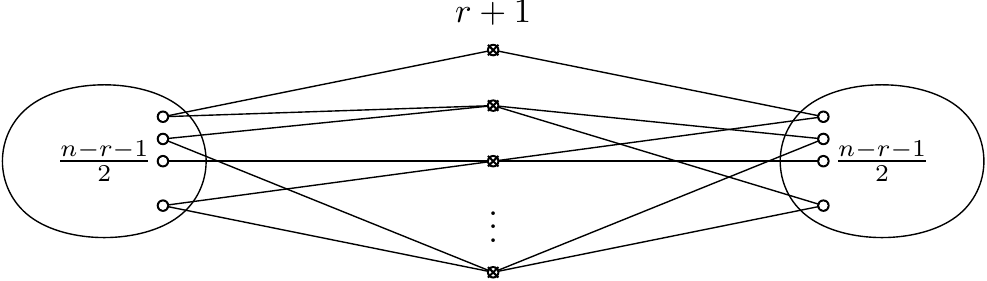}
  \end{center}
  \caption{In an $r$-fault-tolerant spanner, removing $r+1$ vertices may
  disconnect the graph in such a way that no component has size greater
  than $n/2$.}
  \figlabel{rft-problem}
\end{figure}

\subsection{Robustness and Magnification}

A function $h$ is called a \emph{magnification (or vertex-expansion)
function} \cite[Page~390]{kalai91}, for the graph $G=(V,E)$ if, for all
$S\subseteq V$,
\[
    |N(S)| \ge h(|S|) \enspace ,
\]
where $N(S)$ denotes the set of vertices in $V\setminus S$ that
are adjacent to vertices in $S$.  Of particular interest are graphs
that have a magnification function $h(x)=cx$, for fixed $c>1$, and
all $x\in\{1,\ldots,|V|/2\}$.  Such graphs are called \emph{vertex
expanders}, and have a long history and an enormous number of applications
\cite{hlw06}.

If $G$ is $f(k)$-robust, then there exists a magnification function,
$h$, for $G$ that satisfies $h(x) \ge k$, for all $x> f(k)-k$
and every $k\in\{1,\ldots,\lfloor\max\{k':f(k')\le n/2\}\rfloor\}$;
see \figref{magnification}.  This can be proven by contradiction:
If $h(x)$ must be less than $k$ for some $x> f(k)-k$, then there
exists a set $S''$ of size $x>f(k)-k$ such that $|N(S'')|< k$.  Taking
$S=N(S'')\cup\{x_1,\ldots,x_{k-|N(S'')|}\}$, where each $x_i$ is chosen
arbitrarily from $V\setminus N(S'')$ yields a set, $S$, of size $k$, such
that $G\setminus S$, has no component of size greater than $n-x < n-f(k)$.

\begin{figure}
  \begin{center}
    \includegraphics[width=.95\linewidth]{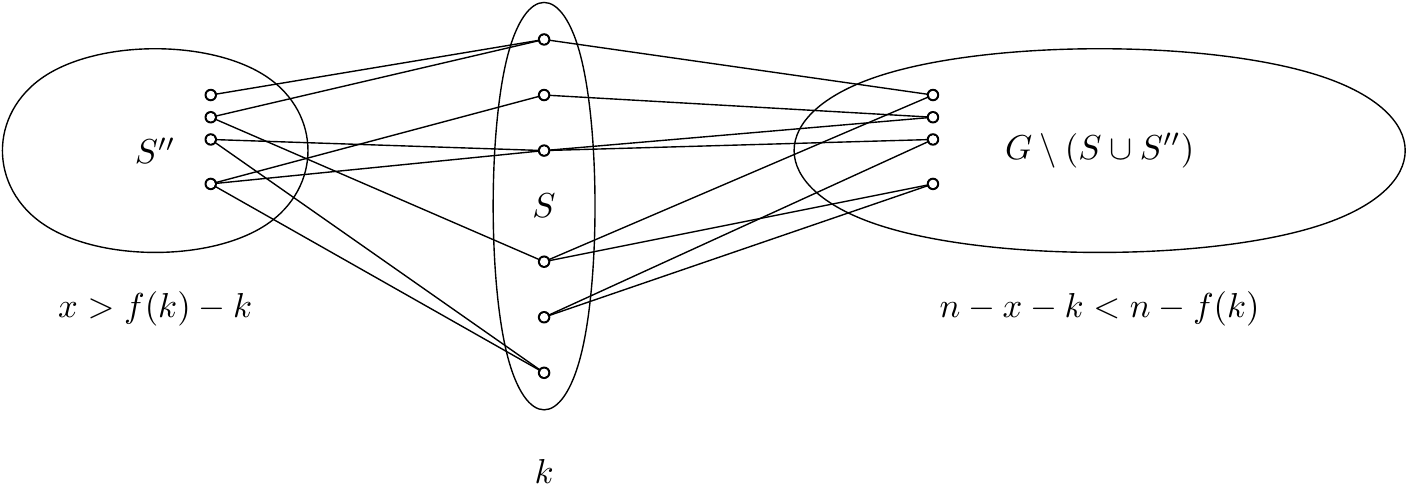}
  \end{center}
  \caption{If $G$ does not have a magnification function, $h$, with
  $h(x)\ge k$, for all $x> f(k)-k$, then $G$ is not $f(k)$-robust.}
  \figlabel{magnification}
\end{figure}

If we think of $f(k)$ as a continuous increasing function (and
hence invertible) then the above argument says that any $f(k)$-robust
spanner with $n$ vertices has a magnification function $h(x)$ such that
$h(x)\in\Omega(\min\{n-x,f^{-1}(x)\})$.  This implies, for example, that the
smallest separator in an $f(k)$-robust spanner with $n$ vertices has
size $\Omega(f^{-1}(n/2))$.

Unfortunately, achieving $f(k)$-robustness is considerably more difficult
than just obtaining a magnification function of the preceding form;
there exist vertex expanders with a linear number of edges \cite{hlw06},
so they have magnification functions of the form $h(x)=cx$, with
fixed $c>1$.  However, these graphs can not be $f(k)$-robust since,
in \thmref{general-lower-bound-1d}, we show that $f(k)$-robust spanners
have a superlinear number of edges, for any function $f(k)$.

\subsection{Overview of Results}

In this paper, we prove upper and lower bounds on the size (number of
edges) needed to achieve $f(k)$-robustness.
These bounds are expressed as a dependence on the function $f(k)$.
In particular, the number of edges depends on the function $f^*(n)$,
which is the maximum number of times one can iterate the function $f$ on
an initial input $k_0$ before exceeding $n$.  As a concrete example, if
$f(k)=2k$, then  $f^*(n)=\floor{\log_2 n}$ (with the initial input $k_0=1$).

Our most general lower-bound, \thmref{general-lower-bound-1d}, states
that, for any constant, $t>1$, there exists one-dimensional point
sets of size $n$ for which any $f(k)$-robust $t$-spanner has size
$\Omega(nf^*(n))$.  For one-dimensional point sets, we can almost match
this lower-bound: \thmref{general-1d} states that any one-dimensional
point set of size $n$ has an $O(f(k)f^*(k))$-robust 1-spanner of size
$O(nf^*(n))$.  Furthermore, if $f(k)$ is sufficiently fast-growing, this
construction is $O(f(k))$-robust, and hence has optimal size.  For point
sets in dimension $d>1$, our upper and lower bounds diverge by a factor
of $k$.  \thmref{dd} shows that, for any set of $n$ points in $\R^d$
and any fixed $t>1$, there exists an $O(kf(k))$-robust $t$-spanner of
size $O(nf^*(n))$.

As a concrete example, we can consider a function $f(k)\in O(k^2)$.
Removing any set $S$ of vertices from a $n$ vertex $O(k^2)$-robust
$t$-spanner leaves a set of at least $n-O(|S|^2)$ vertices which continue
to have $t$-spanning paths between them.  Our results show that, in
one dimension, $O(k^2)$-robust spanners can be constructed that have
$O(n\log\log n)$ edges and this is optimal.  In two and higher dimensions,
$O(k^2)$-robust spanners can be constructed that have $O(n\log n)$ edges.

The remainder of the paper is organized as follows:  \Secref{one-d}
gives results for 1-dimensional point sets, \secref{d-d} gives results
for $d$-dimensional point sets, and \secref{summary} summarizes and
concludes with directions for further research.

\section{One-Dimensional Point Sets}
\seclabel{one-d}

In this section, we consider constructions of robust $t$-spanners for
1-dimensional point sets.  Throughout this section $V=\{x_1,\ldots,x_n\}$
is a set of real numbers with $x_1<x_2<\cdots<x_n$.  We begin by giving
a construction of an $O(k\log k)$-robust 1-spanner having $O(n\log n)$
edges.  This construction contains most of the ideas needed for the
construction of $O(f(k))$-robust 1-spanners for more general $f$.

\subsection{An $O(k\log k)$-robust spanner with $O(n\log n)$ edges}

We now consider the following graph, $G_{2\times}=(V,E)$ which is
closely related to the hypercube.   The edge set, $E$, of $G_{2\times}$
consists of
\[
  E = \{x_ix_{i+2^j} : j\in\{0,\ldots,\floor{\log n}\},\, 
        i\in\{1,\ldots,n-2^j\} \} \enspace .
\] 
Notice that $G_{2\times}$ is a 1-spanner since it contains every
edge of the form $x_ix_{i+1}$, for $i\in\{1,\ldots,n-1\}$. Furthermore,
$G_{2\times}$ has size $O(n\log n)$ since every vertex has degree at most
$2\floor{\log n}+2$.  We now prove an upper-bound on the robustness of
$G_{2\times}$ by using the probabilistic method.

\begin{thm}\thmlabel{better-upper-bound-1d}\thmlabel{klogk-1d}
  Let $V\subset \R$ be any set of $n$ real numbers.  Then there exists
  an $O(k\log k)$-robust $1$-spanner of $V$ of size $O(n\log n)$.
\end{thm}

\begin{proof}
  Let $S$ be any non-empty subset of $V$ and let $k=|S|$.  Select a
  random integer $r\in\{0,1,2,3,\ldots,2^{\ceil{\log n}}-1\}$ and
  consider the subgraph, $G'$, of $G_{2\times}$ consisting only of
  the edges of the form $x_ix_{i+2^j}$ where $i-r\equiv 0\pmod{2^j}$.
  One can think of the edges of $G'$ as a set $O(\log n)$ monotone paths
  that all contain $x_r$; one of these paths contains every vertex in $V$,
  another contains every second vertex, yet another contains every fourth
  vertex, and so on.  (For readers with a background in data structures,
  $G'$ looks a lot like a perfect skiplist in which $x_r$
  appears at the top level; see \figref{g-prime}.)

  \begin{figure}
    \begin{center}
      \includegraphics[width=.95\linewidth]{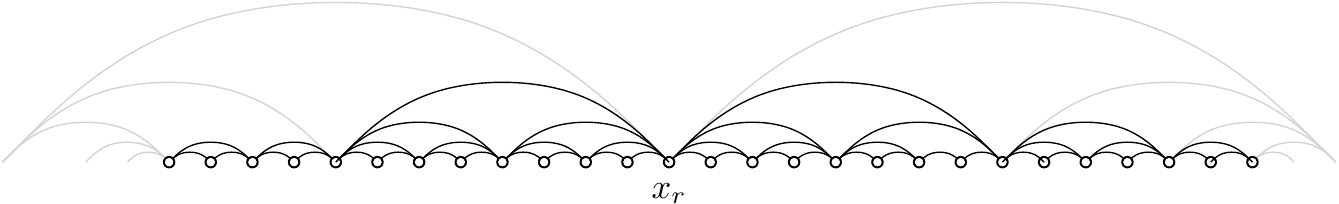}
    \end{center}
    \caption{The graph $G'$}
    \figlabel{g-prime}
  \end{figure}

  For a vertex $x_i\in S$, let $j=j(i)$ be the largest integer such that
  $i-r$ is a multiple of $2^j$.  Then we say that $x_i$ \emph{kills}
  the vertices $x_{i-2^{j}+1},\ldots,x_{i+2^{j}-1}$ in $G'$; see
  \figref{killing}.  When this happens, the \emph{cost} of $x_i$ is
  $c(x_i)=2^{j+1}-1$, which is the number of vertices killed by $x_i$.
  Observe that, unless $i<2^{j}$ or $i>n-2^{j}$, $G'$ contains the
  edge $x_{i-2^{j}}x_{i+2^{j}}$ that ``jumps over'' all the vertices
  killed by $x_i$.  Therefore, if we define $S^+$ to be the set of all
  vertices killed by vertices in $S$, then $G'\setminus S$ (and hence
  also $G_{2\times}\setminus S$) is a 1-spanner of $V\setminus S^+$; it
  contains a path that visits all vertices of $V\setminus S^+$ in order.
  
  \begin{figure}
    \begin{center}
      \includegraphics[width=.95\linewidth]{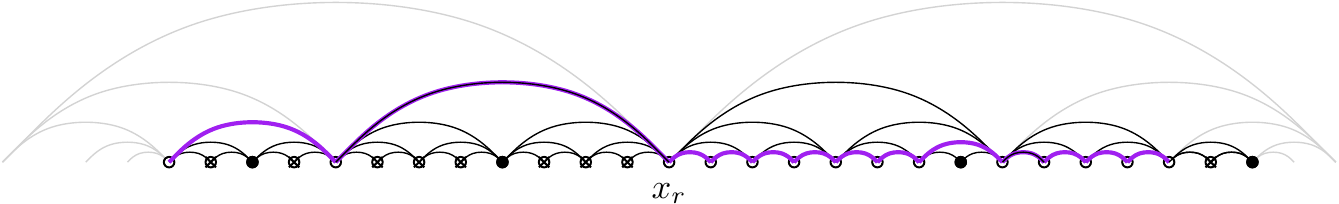}
    \end{center}
    \caption{Constructing the set $S^+$ (whose elements are denoted
    by $\times$ and \textbullet)
      from the set $S$ (whose elements are denoted by \textbullet).}
    \figlabel{killing}
  \end{figure}
  
  We say that a vertex $x\in S$ is \emph{cheap} if $c(x) < 4k$ and
  \emph{expensive} otherwise.  We call our choice of $r$ a \emph{failure}
  if
  \begin{enumerate}
    \item $\mathcal{A}$: some vertex of $S$ is expensive; or
    \item $\mathcal{B}$: the total cost of all cheap vertices exceeds
      $4k\log k+12k$
  \end{enumerate}
  We declare our choice of $r$ a \emph{success} if neither
  $\mathcal{A}$ nor $\mathcal{B}$ holds.  Observe that, in the case
  of a success, we obtain a set $S^+$, $|S^+|\in O(k\log k)$, such that
  $G_{2\times}\setminus S$ is a 1-spanner of $V\setminus S^+$.  Therefore,
  all that remains is to show that the probability of success is greater
  than 0.
  
  We first note that the probability any particular $x_i\in S$ is
  expensive is at most $1/4k$.  This is because $x_i$ is expensive if and
  only if $(i-r)\equiv 0 \pmod {2^{\ceil{\log(4k)}}}$.  The probability
  of selecting $r$ with this property is only $1/2^{\ceil{\log(4k)}}
  \le 1/4k$. Therefore, by the union bound,
  \[
     \Pr\{\mathcal{A}\} \le k/4k = 1/4 \enspace .
  \]

  To upper-bound the total expected cost of cheap vertices, we note that,
  if $x_i\in S$ kills $2^{j+1}-1$ vertices, then $i-r\equiv 0\pmod{2^j}$.
  The probability that this happens is $1/2^{j}$. 
  Letting $S^{c}$
  denote the set of cheap vertices in $S$, the total expected cost of
  all cheap vertices is at most
  \begin{eqnarray*}
     \E\left[\sum_{x\in S^{c}}c(x) \right] 
     &\le& \E\left[\sum_{x\in S}\min\{2^{\floor{\log 4k}},c(x)\} \right]  \\
    &\le&  k\sum_{j=0}^{\floor{\log (4k)}} (2^{j+1}-1)/2^j \\
     &\le& k\sum_{j=0}^{\floor{\log (4k)}} 2 \\
     &\le & 2k\log k + 6k \enspace .
  \end{eqnarray*}
  Therefore, by Markov's Inequality, $\Pr\{\mathcal{B}\}\le 1/2$.
  By the union bound
  \[
     \Pr\{\mbox{$\mathcal{A}$ or $\mathcal{B}$}\} \le 1/4+ 1/2 < 1
  \enspace . \qedhere
  \]
\end{proof}

\subsection{A General Construction}
\seclabel{iterated}

Let $k_0\ge 1$ be a constant and let $f:\R\rightarrow\R$ be any function
that is convex, increasing over the interval $[k_0,\infty)$, and such that
$f(k_0+1)-f(k_0) > 1$.  Let $f^{i}(k)$ be the function $f$ iterated $i$
times on the initial value $k$, i.e.,
\[
   f^{i}(k) = \underbrace{f(f(f(\cdots f}_{i}(k)\cdots))) \enspace .
\]
We use the convention that $f^0(x) = k_0$ for all $x$.  We define the
\emph{iterated $f$-inverse function}
\[
   f^*(n) = \max\{i : f^{i}(k_0) \le n\} \enspace .
\] 
Notice that, for any $k> k_0$, there exists $i$ such that
\[
   f^i(k_0) < k \le f(f^i(k_0)) \enspace .
\]
In particular, the sequence $k_0,f(k_0),f^2(k_0),\ldots,$ contains
a value $f^{i+1}(k_0)$ such that
\[
      k  \le  f^{i+1}(k_0) < f(k) \enspace .
\]
Another important property is that, since $f(k)$ is increasing, convex,
and $f(k_0+1)-f(k_0)>1$, the function $f(x)/x$ is non-decreasing
for $x\ge k_0$: For every $\delta\ge 0$, and every $x\ge k_0$,
$f(x+\delta)/(x+\delta) \ge f(x)/x$.

For a positive number $x$, we define $\upen{x}=2^{\ceil{\log x}}$, as
the smallest power of 2 greater than or equal to $x$.  From the function,
$f$, we define the graph $G_f=(V,E_f)$ to have the edge set:
\begin{eqnarray*}
    E_f &=& \quad\left\{ x_ix_{i+1} : i\in\{1,\ldots,n-1\} \right\} \\
     && {} \cup \left\{ x_{i}x_{i+\upen{f^j(k_0)}} : j\in\{0,\ldots,f^*(n)\},\,
        i\in\{1,\ldots,n-\upen{f^j(k_0)}\} \right\}
\end{eqnarray*}
The graph $G_f$ clearly has $O(nf^*(n))$ edges.  The following theorem
shows that this graph is a robust spanner:

\begin{thm}\thmlabel{general-1d}
  Let $f$, $f^*$, $k_0$, and $G_f$ be defined as above.  Then the graph
  $G_f$ has $O(nf^*(n))$ edges and is
  \begin{enumerate}
    \item an $O(f(4k)f^*(k))$-robust 1-spanner; and 
    \item an $O(f(4k))$-robust 1-spanner if $f(k)\in k2^{\Omega(\sqrt{\log k})}$.
  \end{enumerate}
\end{thm}

\begin{proof}
  The proof is very similar to the proof of \thmref{klogk-1d}.  Let $S$ be
  any non-empty subset of $V$ and let $k=|S|$.   Select a random integer
  $r$ from the set $\{0,1,2,3,\ldots,\upen{f^{f^*(n)+1}(k_0)}-1\}$.  We consider the
  subgraph $G'$ of $G_f$ that contains only the edges $x_ix_{i+\ell}$
  where $i-r\equiv 0\pmod{\ell}$.  We say that an edge $x_ix_{i+\ell}$
  has \emph{span} $\ell$.

  For an integer $i$, let $j=j(i)$ be the smallest integer such that
  $i-r\not\equiv 0\pmod{\upen{f^j(k_0)}}$; see \figref{spanjump}.
  Informally, if $G'$ has any edge that jumps over $x_i$, then it has an
  edge of span $\upen{f^j(k_0)}$ that jumps over $x_i$.  Then we say that
   $x_i$ \emph{kills} $x_{i-p+1},\ldots,x_{i+q-1}$ where
  \begin{eqnarray*}
     p&=& \left((i-r) \bmod \upen{f^j(k_0)}\right) \text{ and } \\
     q&=& \left((i-r) \bmod \upen{f^j(k_0)}\right) \enspace .
  \end{eqnarray*}
  \begin{figure}
    \begin{center}
      \includegraphics{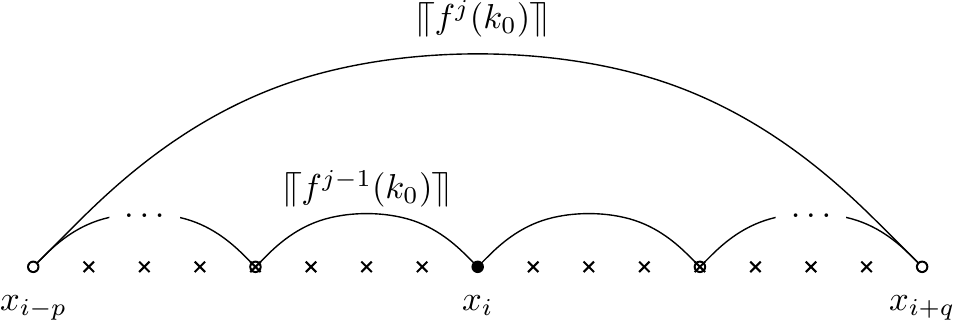}
    \end{center}
    \caption{The vertices killed by $x_i$.}
    \figlabel{spanjump}
  \end{figure}
  As before, we define $S^+$ to be the set of all vertices killed by
  vertices in $S$.  It is easy to verify, since all edges have
  spans that are powers of 2, that the graph $G'\setminus S^+$ (and hence
  also $G_f\setminus S$) contains a path that visits all the vertices of
  $V\setminus S^+$ in order.  Therefore, $G_f\setminus S$ is a 1-spanner
  of $V\setminus S^+$.

  What remains is to show that, with some positive probability, $S^+$
  is sufficiently small to satisfy the appropriate condition, 1 or 2,
  of the theorem.  Define $c(x_i)$ as the number of vertices killed by
  $x_i$.  We say that $x_i$ is \emph{expensive} if $c(x_i) > f(4k)$ and
  \emph{cheap} otherwise.  If $x_i$ is expensive, then $f^{j(i)-1}(k_0)\ge
  4k$ and $i-r\equiv 0 \pmod{\upen{f^{j(i)-1}(k_0)}}$.  Therefore,
  the probability that $x_i$ is expensive is at most $1/f^{j(i)-1}(k_0)
  \le 1/4k$.  Therefore, by the union bound, the probability that $S$
  contains some expensive vertex is at most $1/4$.  All that remains
  is to bound the expected cost of all cheap vertices. Letting $S^c$
  denote the set of cheap vertices in $S$, we obtain
  \begin{align*}
     \E\left[\sum_{x\in S^c} c(x)\right] 
      & \le  k \sum_{j=0}^{f^*(4k)} \upen{f^{j+1}(k_0)}/\upen{f^j(k_0)} \\
      & \le  2k \sum_{j=0}^{f^*(4k)} f^{j+1}(k_0)/f^j(k_0) \\
      & =  2k \sum_{j=0}^{f^*(4k)} f(f^{j}(k_0))/f^j(k_0) \\
      & \le  2k \sum_{j=0}^{f^*(4k)} f(4k)/4k 
           & \text{(since $f(x)/x$ is non-decreasing)} \\
      & \le  (1/2)(f(4k)(f^*(4k)+1) \enspace .
  \end{align*}
  Again, Markov's Inequality implies that the probability that the total
  cost of all cheap vertices exceeds $f(4k)(f^*(4k)+1)$ is at most $1/2$.
  Therefore, the probability of finding a set $S^+$ of size at most
  $f(4k)(f^*(4k)+1)$ is at least
  \[  
     1 - 1/2 - 1/4 > 0 
  \]
  which proves the existence of such a set $S^+$.
  
  To prove the second part of the theorem, we proceed exactly the same
  way, except that the sequence $f^{j+1}(k_0)/f^j(k_0)$,
  $j=0,1,2,\ldots$, becomes
  geometric,\footnote{This is most easily seen by taking $f(k)
  = k\delta^{2\sqrt{(\log k)/(\log\delta)}+1}$.  Then it is
  straightforward to verify that $f^j(\delta) = \delta^{(j+1)^2}$, so
  that $f^{j+1}(\delta)/f^j(\delta)= \delta^{2j+3}$, so the sequence
  is exponentially increasing.  Taking $\delta = 1+\epsilon$ for a
  sufficiently small $\epsilon>0$ allows us to lower-bound any function
  $f(k)\in k2^{\Omega(\sqrt{\log k})}$ this way.} so it is dominated by
  its last term.  This yields:
  \begin{align*}
  \E\left[\sum_{x\in S^c} c(x)\right] 
      & \le  k \sum_{j=0}^{f^*(4k)} \upen{f^{j+1}(k_0)}/\upen{f^j(k_0)} \\
      & \le  2k \sum_{j=0}^{f^*(4k)} f^{j+1}(k_0)/f^j(k_0) \\
      & \le  2ck\left(\frac{f(f^{f^*(4k)}(k_0))}{f^{f^*(4k)}(k_0)}\right) 
            & \text{(for some $c$, since the sum is geometric)} \\
      & \le  2ck\left(\frac{f(4k)}{4k}\right) 
            & \text{(since $f(x)/x$ is non-decreasing)} \\
      & \le  (c/2)f(4k) \enspace , & 
  \end{align*}
  as required.
\end{proof}

Applying \thmref{general-1d} with different functions $f(k)$ yields the
following results.
\begin{cor}
  For any set $V$ of $n$ real numbers, and any constant $\eps >0$,
  there exist $f(k)$-robust 1-spanners $G=(V,E)$ with
  \begin{enumerate}
    \item $f(k)\in O(k\log k)$ and $O(n\log n)$ edges;
    \item $f(k)\in O(k(1+\eps)^{\sqrt{\log k}})$ and $O(n\sqrt{\log n})$
      edges; and
    \item $f(k)\in O(k^{1+\eps})$ and $O(n\log\log n)$ edges.
  \end{enumerate}
\end{cor}

\subsection{Lower Bounds}

In this section, we give lower-bounds on the number of edges in
$f(k)$-robust $t$-spanners.  These lower-bounds hold already for a
specific 1-dimensional point set (the $1\times n$ grid), therefore they
apply to all dimensions $d\ge 1$.

\subsubsection{A Lower Bound for Linear Robustness}

We begin by focusing on the hardest case, $f(k) \in O(k)$.

\begin{thm}\thmlabel{simple-lower-bound-1d}
  Let $V=\{1,\ldots,n\}$ and let $t\ge 1$ be a constant.  Then any
  $O(k)$-robust $t$-spanner of $V$ has $\Omega(n\log n)$ edges.
\end{thm}

\begin{proof}
  To simplify the following discussion, we will assume that $G=(V,E)$ is
  a $ck$-robust $t$-spanner.  Note that we have gone from $O(k)$-robust
  in the statement of the theorem to $ck$-robust in the proof.  This does
  not cause a problem so long as we only consider values of $k$ greater
  than some constant $k_0$ hidden in the $O$ notation.

  We claim that for every natural number $k$ divisible by 4 and every
  $i\in\{ck+1,\ldots,n-ck-1\}$, $G$ has at least $k/2$ \emph{good} edges,
  $xy$, such that $x < i-k/4 < i+k/4 < y$ and such that $y-x \le 2ctk$.

  \begin{figure}
    \begin{center}\includegraphics{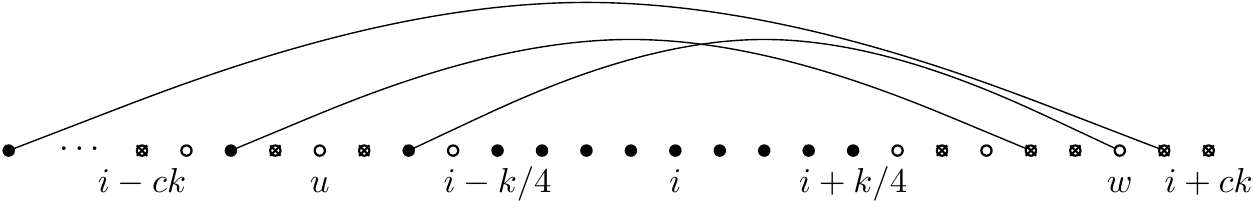}\end{center}
    \caption{After removing $S$ (denoted by \textbullet), there are
      still two vertices $u,w\in V\setminus S^+$ such that $\|uw\|\le 2ck$
      but $\|uw\|_{G\setminus S} > 2ctk$.}
    \figlabel{lower-bound}
  \end{figure}
  To see why the preceding claim is true, consider the set $S$ that
  contains $\{i-k/4,\ldots,i+k/4\}$ as well as the left endpoint of each
  good edge (see \figref{lower-bound}).  The set $S$ has size at most
  $k$ and, in $G\setminus S$, the only edges $xy$ with $x<i$ and $y>i$
  have length greater than $2ctk$.  Now consider any $S^+\supseteq S$,
  with $|S^+|\le ck$.  Since $|S^+|\le ck$ there is at least one element
  $u\in\{i-ck,\ldots,i-1\}$ that is not in $S^+$ and at least one element
  $w\in\{i+1,\ldots,i+ck\}$ that is not in $S^+$.  Now,
  \[    w-u \le 2ck \]
  and, in $G\setminus S$, every path from $u$ to $w$ uses an edge of length
  greater than $2ctk$.  Therefore,
  \[
     \frac{\|uw\|_{G\setminus S}}{\|uw\|} > \frac{2ctk}{2ck} = t \enspace .
  \]
  This contradicts the assumption that $G$ is $ck$-robust $t$-spanner, so we
  conclude that there are, indeed at least $k/2$ good edges.
  
  Applying the above argument to $i=ck+j2ctk$, for
  $j\in\{0,\ldots,\floor{(n-ck)/2ctk}\}$ implies that $G$ contains
  $\Omega(n/tc)=\Omega(n)$ edges whose length is in the range
  $[k/2+1,2ctk]$.  Applying this argument for $k\in\{\ceil{(4tck_0)^{j}}
  : j\in\{0,\ldots,\floor{(\log n)/\log(2tck_0)}\}$ proves that, for
  any constants $c,k_0,t>1$, $G$ has $\Omega(n\log n)$ edges.
\end{proof}

\subsubsection{A General Lower Bound}\seclabel{generalized}

Using the iterated functions from \secref{iterated}, we obtain a whole
class of lower-bounds.

\begin{thm}\thmlabel{general-lower-bound-1d}
  Let $k_0$, $f$, and $f^*$ be defined as in \secref{iterated},
  let $V=\{1,\ldots,n\}$, and let $t\ge 1$ be a constant.  Then any
  $f(k)$-robust $t$-spanner of $V$ has $\Omega(nf^*(n))$ edges.
\end{thm}

\begin{proof}
  The proof is similar to the proof of \thmref{simple-lower-bound-1d}.
  We need only consider $f(k)\in\omega(k)$ since, otherwise we can apply
  \thmref{simple-lower-bound-1d}.  We group the edges of the graph
  $G$ into $\Omega(f^*(n))$ classes and show that each class contains
  $\Omega(n)$ vertices.  

  In particular, using the same argument one can show that, for
  any $i\in\{1,\ldots,f^*(n/t)\}$, any $f(k)$-robust $t$-spanner
  of $V$ has $\Omega(n)$ edges whose lengths are in the range
  $[f^i(k_0)/2,2tf^i(k_0)]$.  Since $f(k)$ is superlinear, there exists
  a constant $i_0$ such that, for any $i>i_0$, $f^{i+1}(k_0)/2 > 2tf^i(k_0)$.
  Thus, the number of edges in any $f(k)$-robust $t$-spanner of $V$ is at least
  \[  \Omega(n) \times (f^*(n/t)-i_0) = \Omega(n f^*(n)) \enspace . \qedhere \]
\end{proof}

\begin{cor}\corlabel{lower-bound}
  Let $V=\{1,\ldots,n\}$ and let $c>1$ and $t>1$ be constants.  Then any
  $f(k)$-robust $t$-spanner with
  \begin{enumerate}
    \item $f(k)\in O(k\log k)$ has $\Omega(n\log n/\log\log n)$ edges;
    \item $f(k)\in O(kc^{\sqrt{\log k}})$ has $\Omega(n\sqrt{\log n})$
      edges; and
    \item $f(k)\in O(k^{c})$ has $\Omega(n\log\log n)$ edges.
  \end{enumerate}
\end{cor}

Note that the lower bounds in Parts~2 and 3 of this corollary match the
corresponding upper-bounds while the lower-bound in Part~1 is off by a
factor of $\log\log n$.

\begin{rem}
  The dependence of our lower bounds on the value of $t$ is not
  given in the statements of Theorems~\ref{thm:simple-lower-bound-1d}
  and \ref{thm:general-lower-bound-1d} or in \corref{lower-bound}.
  However, it is readily extracted from their proofs.  In
  \thmref{simple-lower-bound-1d}, each value of $k$ shows the existence
  of $\Omega(n/t)$ edges and there are $\Omega(\log_t n)$ values of $k$,
  so the lower-bound is $\Omega((n\log n)/(t\log t))$.

  In \thmref{general-lower-bound-1d}, each value of $k$ shows the
  existence of $\Omega(n/t)$ edges, but now the number of values of $k$
  is $f^*(n/t) - f^*(x_0)$ where $x_0$ is the minimum value such that
  $f(x_0)\ge 4tx_0$. (Informally, $x_0$ is where the slope of $f$ exceeds
  $4t$.)  Thus, in \thmref{general-lower-bound-1d}, the lower-bound is
  $\Omega((n/t)(f^*(n)-f^*(x_0))))$.  It is fairly straightforward to
  apply this bound to the choices of $f$ used in \corref{lower-bound}
  or to other choices of $f$.  For example, applying it to Case~3 of
  \corref{lower-bound} we get $x_0=\Theta(t^{1/(c-1)})$ and the result
  that any $O(k^c)$-robust $t$-spanner has $\Omega((n/t)(\log\log n -
  \log\log t - \log (1/(c-1))))$ edges.
\end{rem}

\section{Higher Dimensions}
\seclabel{d-d}

In this section, we give a family of constructions for point sets
$V\subset\R^d$, $d\ge 1$.  These constructions make use of dumbbell
tree spanners \cite[Chapter~11]{ns07}.  In particular, they make use
of binary dumbbell trees, first used by Arya \etal\ \cite{admss95} in
the construction of low-diameter spanners.  A full description of the
construction (and proof of existence) of binary dumbbell trees can be
found in the notes by Smid \cite{s12}.

A \emph{(binary) dumbbell tree spanner} of $V$ is defined by a set of
$O(1)$ binary trees $\mathcal{T}=\{T_1,\ldots,T_p\}$, each having $n$
leaves.  Each node, $u$, in each of these trees is associated with one
element, $r(u)\in V$.  For each $i\in\{1,\ldots,p\}$, and each $x\in V$,
$T_i$ contains exactly one leaf, $u$, such that $r(u)=x$ and at most
one internal node, $w$, such that $r(w)=x$.
For any two points $x,y\in V$, there exists some tree, $T_i$,
with two leaves, $u$ and $v$, such that $r(u)=x$, $r(v)=y$ and
the path, $u,\ldots,v$ in $T_i$ defines a path $r(u),\ldots,r(v)$
whose Euclidean length is at most $t'\|xy\|$, where $t'>1$ is a
parameter in the construction of the dumbbell tree.  Thus, the graph
$G_\dumbbell=(V,E_\dumbbell)$ obtained by taking
\[
   E_\dumbbell = \bigcup_{i=1}^p\{r(u)r(v) : \text{$uv$ is an edge of $T_i$} \}
\]
is a $t'$-spanner of $V$.  

The size (number of edges) of a dumbbell tree spanner is clearly
$O(pn)=O(n)$.  For a fixed dimension, $d$, as a function of $t$ and as $t$
approaches 1, the number of trees, $p$, is $O(\log(1/(t-1))/(t-1)^d)$.
In particular, for $t=1+\eps$, $p\in O(\log(1/\eps)/\eps^d)$.


In the following, we will often treat the nodes of each tree, $T_i$,
in a dumbbell tree decomposition as if the nodes are elements of $V$.
This will happen, for example, when we make statements like ``the path in
$T_i$ from the leaf containing $x$ to the leaf containing $y$ has length
at most $t'\|xy\|$.''  We do this to avoid the cumbersome phraseology
required to distinguish between a node $u\in T_i$ and the node $r(u)\in V$
associated with $u$.  Hopefully the reader can tolerate this informality.

\begin{thm}\thmlabel{dd}
  Let $k_0$, $f$, and $f^*$ be defined as in \secref{iterated} and let
  $d\ge 1$ and $t>1$ be constants.  Let $V\subset \R^d$ be any set of
  $n$ points in $\R^d$.  Then, for any constant $t>1$,  there exists an
  $O(kf(k))$-robust $t$-spanner of $V$ with $O(nf^*(n))$ edges.
\end{thm}

\begin{proof} 
  Fix a value $k'>1$ and recall that, in any binary tree, $T$, with $n$
  nodes, there exists a vertex whose removal disconnects $T$ into at most
  3 components each of size at most $n/2$.  Repeatedly applying this fact
  to any component of size greater than $k'$ yields a set of $O(n/k')$
  vertices whose removal disconnects $T$ into components each of size
  at most $k'$ \cite[Lemma~12.1.5]{ns07}; see \figref{dumbbell-chop}.

  Perform the above decomposition for each of the trees $T_1,\ldots,T_p$
  defining a dumbbell tree $t'$-spanner, $G_\dumbbell$, of $V$ with
  $t'=\sqrt{t}$.  This yields a set, $X$, of $O(n/k')$ vertices
  whose removal disconnects every dumbbell tree into components
  each of size at most $k'$.  Using any of the $k'$-fault-tolerant
  spanner constructions cited in the introduction, we can construct a
  $k'$-fault-tolerant $t'$-spanner for $X$ having $O(k'|X|)=O(n)$ edges.
  Let $G_{k'}=(V,E_{k'})$ denote the graph whose edge set contains
  all edges of the dumbbell spanner $G_\dumbbell$ and all edges of a
  $k'$-fault-tolerant spanner on $X$.

  \begin{figure}
    \begin{center}
      \includegraphics{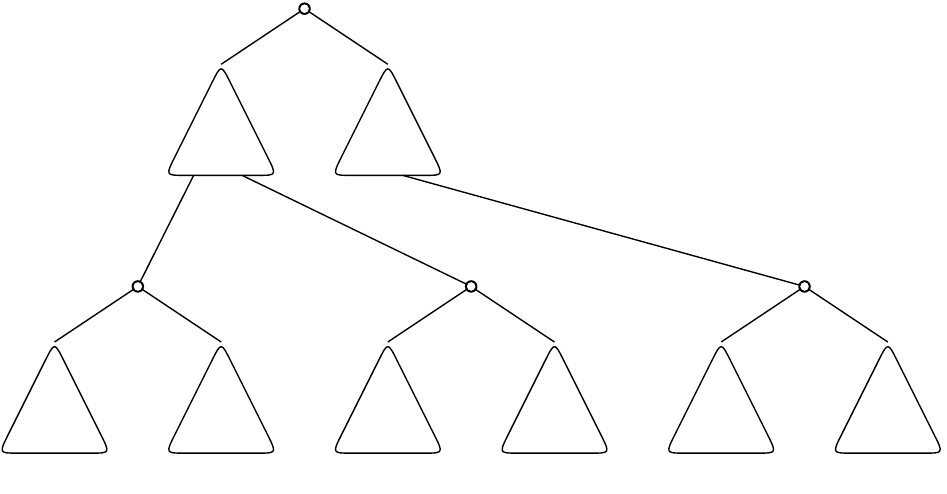}
    \end{center}
    \caption{A dumbbell tree decomposed in components of size $O(k')$
    by the removal of a set $X$ of $O(n/k')$ vertices (each denoted
    by $\circ$).}
    \figlabel{dumbbell-chop}
  \end{figure}

  Suppose that we are now given a set $S\subseteq V$, $|S|=k\le k'$.
  Any vertex $x\in S$ appears at most twice in each tree $T_i$.  For each
  $i\in\{1,\ldots,p\}$, we say that $x$ \emph{kills} all the vertices
  in any component of $T_i\setminus X$ that contains $x$.  Furthermore,
  if $x$ is an element of $X$, then $x$ kills all the vertices in the (at
  most 3) components of $T_i$ whose that have a vertex adjacent to $x$.
  The total number of vertices killed by $x$ is therefore $O(pk')=O(k')$;
  see \figref{dumbbell-kill}.

  \begin{figure}
    \begin{center}
      \includegraphics{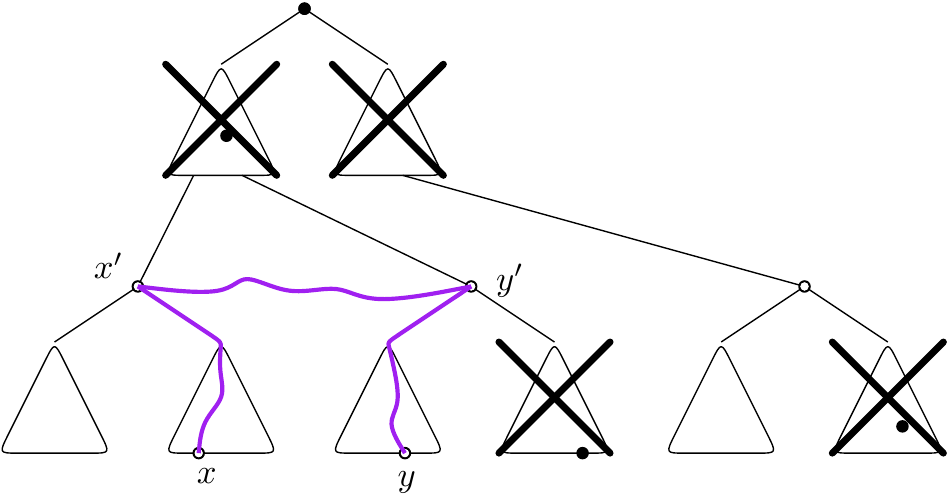}
    \end{center}
    \caption{The set $S$ (whose elements are denoted by \textbullet) kills
      $O(|S|k)$ vertices in each dumbbell tree.}
    \figlabel{dumbbell-kill}
  \end{figure}
  
  Let $S^+$ be the set of all vertices killed by all vertices in $S$.
  The size of $S^+$ is $O(kk')$. Consider some pair of vertices $x,y\in
  V\setminus S^+$.  There exists a tree $T_i$ such that the path, in
  $T_i$, from the leaf containing $x$ to the leaf containing $y$ has
  length at most $t'\|xy\|$.  If $x$ and $y$ are in the same component of
  $T_i\setminus X$ then this path is also a path in $G_{k'}\setminus S$.

  If $x$ and $y$ are in different components of $T_i\setminus X$
  then consider the path from the leaf containing $x$ to the leaf
  containing $y$ in $T_i$.  Let $x'$ denote the first node on this
  path that is in $X$ and let $y'$ denote the last node on this path
  that is in $X$.  The graph $G_{k'}\setminus S$ contains a path, from
  the leaf containing $x$, to $x'$, to $y'$, and then finally to $y$,
  where the path from $x'$ to $y'$ uses the $k'$-fault tolerant spanner;
  see \figref{dumbbell-kill}.  Therefore,
  \begin{align*}
    \|xy\|_{G\setminus S} 
       & \le \|xx'\|_{T_i} + t'\|x'y'\| + \|y'y\|_{T_i} \\
       & \le t'(\|xx'\|_{T_i} + \|x'y'\| + \|y'y\|_{T_i}) \\
       & \le t'(\|xx'\|_{T_i} + \|x'y'\|_{T_i} + \|y'y\|_{T_i}) \\
       &  = t'\|xy\|_{T_i} \\
       & \le (t')^2\|xy\|  \\
       & = t\|xy\| \enspace .
  \end{align*}
  Since this is true for every pair $x,y\in V\setminus S^+$, this means
  that $G_{k'}\setminus S$ is a $t$-spanner of $V\setminus S^+$.

  We have just shown how to construct a graph $G_k$ that has $O(n)$
  edges and is $O(kk')$ robust provided that $|S|\le k'$.  To obtain
  a graph that is $kf(k)$-robust for any value of $k$, we take the
  graph $G$ containing the edges of each $G_{k'}$ for $k'\in\{f^i(k_0) :
  i\in\{0,\ldots,f^*(n)\}\}$.  The graph $G$ has $O(nf^*(n))$ edges.
  For any set $S\in \binom{V}{k}$, we can apply the above argument
  on the subgraph $G_{k'}$ with $k \le k' < f(k)$, to show that $G$
  is $O(kf(k))$-robust.
\end{proof}

Applying \thmref{dd} with different functions $f(k)$ yields the
following results.
\begin{cor}\corlabel{dd}
  For any constants $d >1$, $t>1$, $\eps>0$, and any set $V$ of $n$ points
  in $\R^d$, there exist $f(k)$-robust $t$-spanners $G=(V,E)$ with
  \begin{enumerate}
    \item $f(k)\in O(k^2)$ and $O(n\log n)$ edges;
    \item $f(k)\in O(k^2(1+\eps)^{\sqrt{\log k}})$ and $O(n\sqrt{\log n})$
      edges; and
    \item $f(k)\in O(k^{2+\eps})$ and $O(n\log\log n)$ edges.
  \end{enumerate}
\end{cor}

\begin{rem}
  Note that, like our lower bounds, \thmref{dd} and \corref{dd} do
  not express the relationship between the number of edges and the
  spanning ratio, $t$.  As before, this relationship is not hard
  to work out. The number, $p$, of spanning trees in the dumbbell
  tree spanner is $O(\log(1/\eps)/\eps^d)$, where $\eps=1-\sqrt{t}$.
  Each $k'$-fault-tolerant spanner has $O(n\eps^{d-1})$ edges \cite{l99}
  and we construct one of these for $f^*(n)$ different values of $k'$.
  Thus, the total number of edges in our constructions is
  $O(n(f^*(n)/\eps^{d-1} + \log(1/\eps)/\eps^d))$.
\end{rem}

\subsection{Linear-Size (Kind of) Robust Spanners}

The lower bound in \thmref{general-lower-bound-1d} shows that linear-size
$f(k)$-robust $t$-spanners do not exist for any function $f(k)$.  In this
section, we show that there are linear sized graphs that satisfy a weaker
definition of robustness.

We say that a graph $G=(V,E)$ is \emph{$f(k,n)$-hardy} if, for every
subset $S\subseteq V$, there exists a superset $S^+\supseteq S$,
$|S^+|\le f(|S|,|V|)$, such that $G\setminus S$ is a $t$-spanner of
$V\setminus S^+$.  Note that this definition is almost identical to that
of robustness except that the size of $S^+$ may also depend on $|V|$.  In particular, any $f(k)$-robust $t$-spanner is also an $f'(k,n)$-hardy $t$-spanner with $f'(k,n)=f(k)$.

\begin{thm}\thmlabel{linear-size}
  If $f(k,n)$-hardy $t$-spanners with $O(n\cdot s(n))$ edges exist for
  all $V\subset\R^d$, then $O(f(k,n)\cdot s(n))$-hardy $t$-spanners with
  $O(n)$ edges exist for all $V\subset\R^d$.
\end{thm}

\begin{proof}
  Perform the same dumbbell tree decomposition used in the proof of
  \thmref{dd} to obtain a set $X$ of $O(n/s(n))$ nodes whose removal
  partitions each dumbbell tree into components of size at most $s(n)$.
  Construct an $f(k,n)$-hardy $t$-spanner on the elements of $X$.
  The size of the resulting graph is 
  \begin{align*}
    O(n) + O(|X|\cdot s(|X|)) 
       & = O(n) + O\left(\frac{n}{s(n)}\cdot s\left(\frac{n}{s(n)}\right)\right) \\
       & \le O(n) + O\left(\frac{n}{s(n)}\cdot s(n)\right) \\
       & = O(n) \enspace . 
  \end{align*}
  The same argument used to prove \thmref{dd} shows that the resulting
  construction is $O(f(k,n)s(n))$-hardy.  (Each vertex of $X$ that
  belongs to $S$ results in the loss of at most 3 components in each
  dumbbell tree, each of size at most $s(n)$.)
\end{proof}

The following corollary is obtained by combining \thmref{linear-size}
with some of our upper-bound constructions:
\begin{cor}\corlabel{linear-size}
  For any constant $\epsilon >0$, there exist linear size
  \begin{enumerate}
    \item $O(k\log k\log n)$-hardy $1$-spanners of any
      $V\subset \R$;
    \item $O(k^{1+\eps}\log\log n)$-hardy $1$-spanners of any
      $V\subset \R$;
    \item $O(k^2\log n)$-hardy $t$-spanners of any
      $V\subset \R^d$; and
    \item $O(k^{2+\eps}\log\log n)$-hardy $t$-spanners of any
      $V\subset \R^d$.
  \end{enumerate}
\end{cor}

\begin{rem}
One can use the same argument used to prove
Theorems~\ref{thm:simple-lower-bound-1d} and
\ref{thm:general-lower-bound-1d} to study the hardiness/space tradeoff
in hardy spanners.  For example, one can show that any $f(k)g(n)$-hardy
$t$-spanner of the $1\times n$ grid has $\Omega((nf^*(n)/g(n))$ edges.
This implies, for example, that Part~2 of \corref{linear-size} is tight;
it is not possible to asymptotically reduce the dependence on $k$ or $n$
while keeping a linear number of edges (apply the tradeoff result with
$f(k)\in O(k^{1+\eps})$, and $f^*(n)=g(n)\in \Theta(\log\log n)$).
\end{rem}

\section{Summary}
\seclabel{summary}

We have introduced the notion of $f(k)$-robust $t$-spanners and given
upper and lower-bounds on the number of edges in such spanners.  Our lower
bounds show that, for any $f$, $f(k)$-robust spanners sometimes require
a superlinear number of edges, even in one dimension.  Our 1-dimensional
constructions nearly match this lower-bound except when the function
$f$ is nearly linear.

\paragraph{Open problem: Tighter bounds.}
We understand the situation less clearly in two and higher dimensions.
The lower bounds show that $f(k)$-robust $t$-spanners must have
$\Omega(nf^*(n))$ edges, but we have only been able to obtain
$O(kf(k))$-robust $t$-spanners with $O(nf^*(n))$ edges.  Closing this
gap is the main open problem left by this work.

To gain some intuition about which is closer to the truth, the lower
bound or the upper bound, one can study the $\sqrt{n}\times\sqrt{n}$
grid graph; see \figref{grid}.  An argument similar to the proof of
\thmref{klogk-1d}, based on randomly shifting a quadtree, shows that
this graph is an $O(k^2)$-robust $3$-spanner.  Therefore, the vertices
of the $\sqrt{n}\times\sqrt{n}$ grid admit a linear-size $O(k^2)$-robust
3-spanner.  In contrast, \thmref{general-lower-bound-1d} shows that any
$f(k)$-robust $t$-spanner for the $1\times n$ grid has superlinear size.
This suggests that one dimension is the hardest case:

\begin{conj}
If $f(k)$-robust $t$-spanners with $s_f(n)$ edges exist for all one-dimensional point sets, then $O(f(k))$-robust $t$-spanners with $s_f(n)$ edges exist for all point sets in $\R^d$.
\end{conj}

\paragraph{Open problem: Low weight.}

In many cases, the cost of building a network is more closely related to
the total length (rather than number) of edges.  In these cases, one
attempts to construct a graph whose total edge length is close to that
of the minimum spanning tree of $V$. The same lower-bound argument used
in \thmref{general-lower-bound-1d} shows that, in general, $f(k)$-robust
spanners may require edges whose total length is $\Omega(f^*(n))$ times
that of the minimum spanning tree.  Is there a (nearly) matching upper bound?

\paragraph{Open problem: $O(k)$-robust spanners.}

Another fundamental open problem has to do with the number of edges needed
in an $O(k)$-robust $t$-spanner.  We have no upper-bound better than the
trivial $O(n^2)$ and the only lower-bound is $\Omega(n\log n)$.  This is
true even if we restrict our attention to constructing a $t$-spanner
for the 1-dimensional point set $V=\{1,\ldots,n\}$.

\paragraph{Open problem: Induced spanners.}

Finally, we observe that the one-dimensional constructions of
$f(k)$-robust spanners actually satisfy a property that is slightly
stronger than $f(k)$-robustness:  For each of these, the graph $G\setminus
S^+$ is a $t$-spanner.  In other words, vertices not in $V\setminus
S^+$ are not needed in the short paths between pairs of vertices in
$V\setminus S^+$.  Our $d$-dimensional constructions do not have this
stronger property.  It would be interesting to know if $d$-dimensional
constructions having this stronger property exist.

\section*{Acknowledgement}

This work was partly funded by NSERC and CFI.

The research in this paper was started at the workshop on \emph{Models
of Sparse Graphs and Network Algorithms (12w5004)}, hosted at the
Banff International Research Station (BIRS), February 5--10, 2012.
The authors are grateful to the organizers, \mbox{Nicolas~Broutin},
\mbox{Luc~Devroye}, and \mbox{G\'abor~Lugosi}, the other participants, and the staff
at BIRS, for providing a stimulating research environment.

\bibliographystyle{siam}
\bibliography{giant}

\end{document}